\documentclass[11pt]{article}
\usepackage{setspace}
\usepackage{amsfonts}
\usepackage{amsmath}
\usepackage{amssymb}
\usepackage{amsthm}
\usepackage{enumerate}
\numberwithin{equation}{section}
\newtheorem{theorem}{Theorem}[section]
\newtheorem{lemma}[theorem]{Lemma}

\newtheorem{corollary}[theorem]{Corollary}
\newtheorem{observation}[theorem]{Observation}

\newenvironment{definition}[1][Definition.]{\begin{trivlist}
\item[\hskip \labelsep {\bfseries #1}]}{\end{trivlist}}
\newenvironment{notation}[1][Notation.]{\begin{trivlist}
\item[\hskip \labelsep {\bfseries #1}]}{\end{trivlist}}

\newcommand{\twopartdef}[4]
{
	\left\{
		\begin{array}{lll}
			#1 & \mbox{if } #2 \\
			#3 & \mbox{} #4
		\end{array}
	\right.
}

\begin{document}
\bibliographystyle{plain}
\title{Mutual Dimension\footnote{This research was supported in part by National Science Foundation Grants 0652569, 1143830, and 124705. Part of the second author's work was done during a sabbatical at Caltech and the Isaac Newton Institute for Mathematical Sciences at the University of Cambridge. A preliminary version of part of this material was reported at the 2013 Symposium on Theoretical Aspects of Computer Science in Kiel, Germany.}}
\author{Adam Case and Jack H. Lutz\\
Department of Computer Science\\
Iowa State University\\
Ames, IA 50011 USA}
\date{}
\maketitle

\begin{abstract}
We define the lower and upper \emph{mutual dimensions} $mdim(x:y)$ and $Mdim(x:y)$ between any two points $x$ and $y$ in Euclidean space. Intuitively these are the lower and upper densities of the algorithmic information shared by $x$ and $y$.  We show that these quantities satisfy the main desiderata for a satisfactory measure of mutual algorithmic information. Our main theorem, the \emph{data processing inequality} for mutual dimension, says that, if $f:\mathbb{R}^m \rightarrow \mathbb{R}^n$ is computable and Lipschitz, then the inequalities $mdim(f(x):y) \leq mdim(x:y)$ and $Mdim(f(x):y) \leq Mdim(x:y)$ hold for all $x \in \mathbb{R}^m$ and $y \in \mathbb{R}^t$. We use this inequality and related inequalities that we prove in like fashion to establish conditions under which various classes of computable functions on Euclidean space preserve or otherwise transform mutual dimensions between points.
\end{abstract}
\section{Introduction}

Recent interactions among computability theory, algorithmic information theory, and geometric measure theory have assigned a \emph{dimension} $dim(x)$ and a \emph{strong dimension} $Dim(x)$ to each individual point $x$ in a Euclidean space $\mathbb{R}^n$. These dimensions, which are real numbers satisfying $0 \leq dim(x) \leq Dim(x) \leq n$, have been shown to be geometrically meaningful. For example, the \emph{classical} Hausdorff dimension $dim_H(E)$ of any set $E \subseteq \mathbb{R}^n$ that is a union of $\Pi^0_1$ (computably closed) sets is now known \cite{jLutz03a, jHitc05} to admit the pointwise characterization
\begin{equation*}
dim_H(E) = \displaystyle\sup_{x \in E}dim(x).
\end{equation*}
More recent investigations of the dimensions of individual points in Euclidean space have shed light on connectivity \cite{jLutWei08,jTure11}, self-similar fractals \cite{jLutMay08,jDLMT13}, rectifiability of curves \cite{cGuLuMa06,cRetZhe09,jGuLuMa11}, and Brownian motion \cite{jKjoNer09}.

In their original formulations \cite{jLutz03a,jAHLM07}, $dim(x)$ is $cdim(\{x\})$ and $Dim(x)$ is $cDim(\{x\})$, where $cdim$ and $cDim$ are constructive versions of classical Hausdorff and packing dimensions \cite{bFalc03}, respectively. Accordingly, $dim(x)$ and $Dim(x)$ are also called \emph{constructive fractal dimensions}. It is often most convenient to think of these dimensions in terms of the Kolmogorov complexity characterization theorems
\begin{equation}\label{charTheorem}
dim(x) = \displaystyle\liminf_{r \rightarrow \infty}\frac{K_r(x)}{r},\,Dim(x) = \displaystyle\limsup_{r \rightarrow \infty}\frac{K_r(x)}{r},
\end{equation}
where $K_r(x)$, the Kolmogorov complexity of $x$ at precision $r$, is defined later in this introduction \cite{jMayo02,jAHLM07,jLutMay08}. These characterizations support the intuition that $dim(x)$ and $Dim(x)$ are the lower and upper \emph{densities of algorithmic information} in the point $x$.

In this paper we move the pointwise theory of dimension forward in two ways. We formulate and investigate the \emph{mutual dimensions} --- intuitively, the lower and upper densities of shared algorithmic information --- between two points in Euclidean space, and we investigate the \emph{conservation} of dimensions and mutual dimensions by computable functions on Euclidean space. We expect this to contribute to both computable analysis --- the theory of scientific computing \cite{jBraCoo06} --- and algorithmic information theory.

The analyses of many computational scenarios call for quantitative measures of the degree to which two objects are correlated. In classical (Shannon) information theory, the most useful such measure is the \emph{mutual information} $I(X:Y)$ between two probability spaces $X$ and $Y$ \cite{bCovTho06}. In the algorithmic information theory of finite strings, the (\emph{algorithmic}) \emph{mutual information} $I(x:y)$ between two \emph{individual} strings $x,y \in \{0,1\}^*$ plays an analogous role \cite{bLiVit08}. Under modest assumptions, if $x$ and $y$ are drawn from probability spaces $X$ and $Y$ of strings respectively, then the expected value of $I(x:y)$ is very close to $I(X:Y)$ \cite{bLiVit08}. In this sense algorithmic mutual information is a refinement of Shannon mutual information.

Our formulation of mutual dimensions in Euclidean space is based on the algorithmic mutual information $I(x:y)$, but we do not use the seemingly obvious approach of using the binary expansions of the real coordinates of points in Euclidean space. It has been known since Turing's famous correction \cite{jTuri37} that binary notation is not a suitable representation for the arguments and values of computable functions on the reals. (See also \cite{bKo91,bWeih00}.) This is why the characterization theorems (\ref{charTheorem}) use $K_r(x)$, the \emph{Kolmogorov complexity} of a point $x \in \mathbb{R}^n$ at precision $r$, which is the minimum Kolmogorov complexity $K(q)$ --- defined in a standard way \cite{bLiVit08} using a standard binary string representation of $q$ --- for all rational points $q \in \mathbb{Q}^n \cap B_{2^{-r}}(x)$, where $B_{2^{-r}}(x)$ is the open ball of radius $2^{-r}$ about $x$. For the same reason we base our development here on the \emph{mutual information} $I_r(x:y)$ between points $x \in \mathbb{R}^m$ and $y \in \mathbb{R}^n$ at \emph{precision} $r$. This is the minimum value of the algorithmic mutual information $I(p:q)$ for all rational points $p \in \mathbb{Q}^m \cap B_{2^{-r}}(x)$ and $q \in \mathbb{Q}^n \cap B_{2^{-r}}(y)$. Intuitively, while there are infinitely many pairs of rational points in these balls and many of these pairs will contain a great deal of ``spurious'' mutual information (e.g., \emph{any} finite message can be encoded into both elements of such a pair), a pair of rational points $p$ and $q$ achieving the minimum $I(p:q) = I_r(x:y)$ will only share information that their proximities to $x$ and $y$ force them to share. Sections 3 and 4 below develop the ideas that we have sketched in this paragraph, along with some elements of the fine-scale geometry of algorithmic information in Euclidean space that are needed for our results. A modest generalization of Levin's coding theorem (Theorem \ref{lds coding theorem} below) is essential for this work.

In analogy with the characterizations (\ref{charTheorem}) we define our mutual dimensions as the lower and upper densities of algorithmic mutual information,
\begin{equation}\label{mdimDefs}
mdim(x:y) = \displaystyle\liminf_{r \rightarrow \infty}\frac{I_r(x:y)}{r},\,Mdim(x:y) = \displaystyle\limsup_{r \rightarrow \infty}\frac{I_r(x:y)}{r},
\end{equation}
in section 5. We also prove in that section that these quantities satisfy all but one of the desiderata (e.g., see \cite{jBell62}) for any satisfactory notion of mutual information.

We save the most important desideratum --- our main theorem --- for section 6. This is the data processing inequality for mutual dimension (actually two inequalities, one for $mdim$ and one for $Mdim$). The data processing inequality of Shannon information theory \cite{bCovTho06} says that, for any two probability spaces $X$ and $Y$ and any function $f: X \rightarrow Y$,
\begin{equation}\label{dpiDef}
I(f(X):Y) \leq I(X:Y)
\end{equation}
i.e., the induced probability space $f(X)$ obtained by ``processing the information in $X$ through $f$'' does not share any more information with $Y$ than $X$ shares with $Y$. The data processing inequality of algorithmic information theory \cite{bLiVit08} says that, for any computable partial function $f: \{0,1\}^* \rightarrow \{0,1\}^*$, there is a constant $c_f \in \mathbb{N}$ (essentially the number of bits in a program that computes $f$) such that, for all strings $x \in dom\,f$ and $y \in \{0,1\}^*$,
\begin{equation}\label{adpiDef}
I(f(x):y) \leq I(x:y) + c_f.
\end{equation}
That is, modulo the constant $c_f$, $f(x)$ contains no more information about $y$ than $x$ contains about $y$.

The data processing inequality for points in Euclidean space is a theorem about functions $f: \mathbb{R}^m \rightarrow \mathbb{R}^n$ that are computable in the sense of computable analysis \cite{jBraCoo06,bKo91,bWeih00}. Briefly, an \emph{oracle} for a point $x \in \mathbb{R}^m$ is a function $g_x: \mathbb{N} \rightarrow \mathbb{Q}^m$ such that $|g_x(r) - x| \leq 2^{-r}$ holds for all $r \in \mathbb{N}$. A function $f: \mathbb{R}^m \rightarrow \mathbb{R}^n$ is \emph{computable} if there is an oracle Turing machine $M$ such that, for every $x \in \mathbb{R}^m$ and every oracle $g_x$ for $x$, the function $r \mapsto M^{g_x}(r)$ is an oracle for $f(x)$.

Given (\ref{mdimDefs}), (\ref{dpiDef}), and (\ref{adpiDef}), it is natural to conjecture that, for every computable function $f: \mathbb{R}^m \rightarrow \mathbb{R}^n$, the inequalities
\begin{equation}\label{dim dpiDef}
mdim(f(x):y) \leq mdim(x:y),\,\,\,\,Mdim(f(x):y) \leq Mdim(x:y)
\end{equation}
hold for all $x \in \mathbb{R}^m$ and $y \in \mathbb{R}^t$. However, this is not the case. For a simple example, there exist computable functions $f: \mathbb{R} \rightarrow \mathbb{R}^2$ that are \emph{space-filling}, e.g., satisfy $[0,1]^2 \subseteq range\,f$ \cite{jCoDaMc12}. For such a function $f$ we can choose $x \in \mathbb{R}$ such that $dim(f(x)) = 2$. Letting $y = f(x)$, we then have
\begin{displaymath}
mdim(f(x):y) = dim(f(x)) = 2 > 1 \geq Dim(x) \geq Mdim(x:y),
\end{displaymath}
whence both inequalities in (\ref{dim dpiDef}) fail.

The difficulty here is that the above function $f$ is extremely sensitive to its input, and this enables it to compress a great deal of ``sparse'' high-precision information about its input $x$ into ``dense'' lower-precision information about its output $f(x)$. Many theorems of mathematical analysis exclude such excessively sensitive functions by assuming a given function $f$ to be \emph{Lipschitz}, meaning that there is a real number $c > 0$ such that, for all $x$ and $x'$, $|f(x) - f(x')| \leq c|x - x'|$. This turns out to be exactly what is needed here. In section 6 we prove prove the \emph{data processing inequality} for mutual dimension (Theorem \ref{dpi}), which says that the conditions (\ref{dim dpiDef}) hold for every function $f: \mathbb{R}^m \rightarrow \mathbb{R}^n$ that is computable and Lipschitz. In fact, we derive the data processing inequality from the more general \emph{modulus processing lemma} (Lemma \ref{mod proc lemma}). This lemma yields quantitative variants of the data processing inequality for other classes of functions. For example, we use the modulus processing lemma to prove that, if $f: \mathbb{R}^m \rightarrow \mathbb{R}^n$ is \emph{H\"{o}lder with exponent} $\alpha$ (meaning that $0 < \alpha \leq 1$ and there is a real number $c > 0$ such that $|f(x) - f(x')| \leq c|x - x'|^{\alpha}$ for all $x, x' \in \mathbb{R}^m$), then the inequalities
\begin{equation}
mdim(f(x):y) \leq \frac{1}{\alpha}mdim(x:y),\,\,\,\,Mdim(f(x):y) \leq \frac{1}{\alpha}Mdim(x:y)
\end{equation}
hold for all $x \in \mathbb{R}^m$ and $y \in \mathbb{R}^t$.

In section 7 we derive \emph{reverse} data processing inequalities, e.g., giving conditions under which $mdim(x:y) \leq mdim(f(x):y)$. In section 8 we use data processing inequalities and their reverses to explore conditions under which computable functions on Euclidean space preserve, approximately preserve, or otherwise transform mutual dimensions between points.
\section{Preliminaries}

We write $\mathbb{Z}$ for the set of integers, $\mathbb{N}$ for the set of non-negative integers, $\mathbb{Q}$ for the set of rationals, $\mathbb{R}$ for the set of reals, and $\mathbb{R}^n$ for the set of all $n$-vectors $(x_1, x_2, \cdots, x_n)$ such that each $x_i \in \mathbb{R}$. Our logarithms are in base 2. We denote the cardinality of a set $A$, the length of a string $s \in \{0,1\}^*$, and the distance between two points $x,y \in \mathbb{R}^n$ (using the Euclidean metric) by $|A|$, $|s|$, and $|x - y|$ respectively. We also denote the $i^{th}$ string in $\{0,1\}^*$ by $s_i$.

Our use of Turing machines is strictly limited to self-delimiting (or prefix) machines. Because of this, we refer to a self-delimiting Turing machine simply as a Turing machine. We refer the reader to Li and Vitanyi \cite{bLiVit08} for a detailed explanation of how self-delimiting Turing machines work.

The (\emph{conditional}) \emph{Kolmogorov complexity} of a string $x \in \{0,1\}^*$ \emph{given} a string $y \in \{0,1\}^*$ with respect to a Turing machine $M$ is
\begin{displaymath}
K_M(x\,|\,y) = \min\{|\pi|\,\Big|\,\pi \in \{0,1\}^* \textnormal{ and } M(\pi, y) = x\}.
\end{displaymath}
The \emph{Kolmogorov complexity} of $x$ with respect to $M$ is $K_M(x) = K_M(x\,|\,\lambda)$, where $\lambda$ is the \emph{empty string}. A Turing machine $M'$ is \emph{optimal} if, for every Turing machine $M$, there is a constant $c_M \in \mathbb{N}$ such that, for all $x \in \{0,1\}^*$,
\begin{displaymath}
K_{M'}(x) \leq K_M(x) + c_M.
\end{displaymath}
We call $c_M$ an \emph{optimality constant} for $M$. It is well-known that every universal Turing machine is optimal \cite{bLiVit08}. Following standard practice, we fix a universal, hence optimal, Turing machine $U$; we omit it from the notation, writing $K(x) = K_U(x)$ and $K(x\,|\,y) = K_U(x\,|\,y)$; and we call these the \emph{Kolmogorov complexity} of $x$ and the (\emph{conditional}) \emph{Kolmogorov complexity} of $x$ \emph{given} $y$, respectively.

The \emph{joint Kolmogorov complexity} of two strings $x,y \in \{0,1\}^*$ is
\begin{align*}
K(x,y) = K(\langle x,y \rangle),
\end{align*}
where $\langle \cdot,\cdot \rangle$ is some standard pairing function for encoding two strings. G\'{a}cs \cite{jGacs74} proved the useful identity
\begin{align}\label{kxy}
K(x,y) = K(x) + K(y\,|\,\langle x,K(x) \rangle) + O(1).
\end{align}

The \emph{universal a priori probability} of a set $S \subseteq \{0,1\}^*$ is
\begin{displaymath}
{\bf m}(S) = \displaystyle\sum\limits_{U(\pi) \in S}2^{-|\pi|}.
\end{displaymath}
Since we are using self-delimiting machines, the Kraft inequality tells us that ${\bf m}(\{0,1\}^*) \leq 1$. The \emph{universal a priori probability} of a string $x \in \{0,1\}^*$ is {\bf m}(x) = {\bf m}(\{x\}). For $r \in \mathbb{N}$, we write $K(r)$ for $K(s_r)$ and ${\bf m}(r)$ for ${\bf m}(s_r)$. It is well known that there is a constant $c_0 \in \mathbb{N}$ such that $K(x) \leq |x| + 2\log\,(1+|x|) + c_0$, and hence $K(r) \leq \log\,(1+r) + 2\log(1 + \log\,(1+r)) + c_0$, hold for all $x \in \{0,1\}^*$ and $r \in \mathbb{N}$.

Levin's coding lemma plays an important role in section 3.
\begin{lemma}[coding lemma \cite{jLevi73,jLevi74}]\label{cl}
If $A \subseteq \{0,1\}^* \times \mathbb{N}$ is computably enumerable and satisfies $\Sigma_{(x,l) \in A}2^{-l} \leq 1$, then there is a Turing machine $M$ such that, for each $(x,l) \in A$, there is a string $\pi \in \{0,1\}^l$ satisfying $M(\pi) = x$.
\end{lemma}
\section{Kolmogorov Complexity in Euclidean Space}
We begin by developing some elements of the fine-scale geometry of algorithmic information in Euclidean space. In this context it is convenient to regard the Kolmogorov complexity of a set of strings to be the number of bits required to specify \emph{some} element of the set.

\begin{definition}[Definition] (Shen and Vereshchagin \cite{jSheVer02}).
The \emph{Kolmogorov complexity} of a set $S \subseteq \{0,1\}^*$ is
\begin{displaymath}
K(S) = \min\{K(x)\,|\,x \in S\}.
\end{displaymath}
\end{definition}
Note that $S \subseteq T$ implies $K(S) \geq K(T)$. Intuitively, small sets may require ``higher resolution'' than large sets.

We need a generalization of Levin's coding theorem \cite{jLevi73,jLevi74} that is applicable to certain systems of disjoint sets.

\begin{notation}
Let $B \subseteq \mathbb{N} \times \mathbb{N} \times \{0,1\}^*$ and $r,s \in \mathbb{N}$.
\begin{enumerate}
\item The $(r,t)$-\emph{block} of $B$ is the set $B_{r,t} = \{x \in \{0,1\}^*\,|\,(r,t,x) \in B\}$.
\item The $r^{th}$ \emph{layer} of $B$ is the sequence $B_r = (B_{r,t}\,|\,t \in \mathbb{N})$.
\end{enumerate}
\end{notation}

\begin{definition}
A \emph{layered disjoint system} (LDS) is a set $B \subseteq \mathbb{N} \times \mathbb{N} \times \{0,1\}^*$ such that, for all $r,s,t \in \mathbb{N}$,
\begin{displaymath}
s \neq t \Rightarrow B_{r,s} \cap B_{r,t} = \emptyset.
\end{displaymath}
\end{definition}
Note that this definition only requires the sets within each layer of $B$ to be disjoint.

\begin{theorem}[LDS coding theorem]\label{lds coding theorem} For every computably enumerable layered disjoint system $B$ there is a constant $c_B \in \mathbb{N}$ such that, for all $r,t \in \mathbb{N}$,
\begin{displaymath}
K(B_{r,t}) \leq \log \frac{1}{{\bf m}(B_{r,t})} + K(r) + c_B.
\end{displaymath}
\end{theorem}

\begin{proof}
Assume the hypothesis, and fix a computable enumeration of $B$. For each $r, t \in \mathbb{N}$ such that $B_{r,t} \neq \emptyset$, let $x_{r,t}$ be the first element of $B_{r,t}$ to appear in this enumeration. Let $A$ be the set of all ordered pairs $(x_{r,t}, j+k+2)$ such that $r,t,j,k \in \mathbb{N}$, $B_{r,t} \neq \emptyset$, $k \geq K(r)$, and ${\bf m}(B_{r,t}) \geq 2^{-j}$. It is clear that $A$ is computably enumerable.

For each $r,t \in \mathbb{N}$, let
\begin{displaymath}
j_{r,t} = \min\{j \in \mathbb{N} \, \big| \, {\bf m}(B_{r,t}) > 2^{-j}\},
\end{displaymath}
noting that $j_{r,t} = \infty$ if $B_{r,t} = \emptyset$. For all $r,t \in \mathbb{N}$ such that $B_{r,t} \neq \emptyset$, we have
\begin{align*}
\displaystyle\sum\limits_{\substack{l \in \mathbb{N}\\(x_{r,t}, l) \in A}}2^{-l} &= \displaystyle\sum\limits_{j=j_{r,t}}^{\infty}\,\,\displaystyle\sum\limits_{k=K(r)}^{\infty}2^{-(j+k+2)}\\
&= \displaystyle\sum\limits_{k=K(r)}^{\infty}2^{-(k+1)}\,\,\displaystyle\sum\limits_{j=j_{r,t}}^{\infty}2^{-(j+1)}\\
&= 2^{-K(r)}2^{-j_{r,t}}\\
&< 2^{-K(r)}{\bf m}(B_{r,t}).
\end{align*}
Since the sets in each layer $B_r$ of $B$ are disjoint, it follows that
\begin{align*}
\displaystyle\sum\limits_{(x,l) \in A}2^{-l} &\leq \displaystyle\sum\limits_{r=0}^{\infty}\,\,\displaystyle\sum\limits_{t=0}^{\infty}2^{-K(r)}{\bf m}(B_{r,t})\\
&= \displaystyle\sum\limits_{r=0}2^{-K(r)}\displaystyle\sum\limits_{t=0}^{\infty}{\bf m}(B_{r,t})\\
&= \displaystyle\sum\limits_{r=0}^{\infty}2^{-K(r)}{\bf m}\Bigg(\displaystyle\bigcup_{t=0}^{\infty}B_{r,t} \Bigg)
\end{align*}
\begin{align*}
&\leq \displaystyle\sum\limits_{r=0}^{\infty}2^{-K(r)}{\bf m}(\{0,1\}^*)\\
&\leq \displaystyle\sum\limits_{r=0}^{\infty}2^{-K(r)}\\
&\leq \displaystyle\sum\limits_{r=0}^{\infty}{\bf m}(r)\\
&= {\bf m}(\{0,1\}^*)\\
&\leq 1.
\end{align*}

We have now shown that the set $A$ satisfies the hypothesis of Lemma \ref{cl}. Let $M$ be a Turing machine for $A$ as in that lemma, and let $c_B = c_M + 3$, where $c_M$ is an optimality constant for $M$. To see that $c_B$ affirms the theorem, let $r,t \in \mathbb{N}$ be such that $B_{r,t} \neq \emptyset$. (The theorem is trivial if $B_{r,t} = \emptyset$, since the right-hand side is infinite.) Then $(x_{r,t}, j_{r,t} + K(r) + 2) \in A$, so there is a program $\pi \in \{0,1\}^{j_{r,t} + K(r) + 2}$ such that $M(\pi) = x_{r,t}$. We thus have
\begin{align*}
K(B_{r,t}) &\leq K(x_{r,t})\\
					 &\leq K_M(x_{r,t}) + c_M\\
					 &\leq |\pi| + c_M\\
					 &= j_{r,t} + K(r) + 2 + c_M\\
					 &= \lfloor \log\frac{1}{\bf m}(B_{r,t}) \rfloor + 1 + K(r) + 2 + c_M\\
					 &\leq \log\frac{1}{{\bf m}(B_{r,t})} + K(r) + c_B. \qedhere
\end{align*}
\end{proof}

Note that Levin's coding theorem \cite{jLevi73,jLevi74}, the nontrivial part of which says that $K(x) \leq \log\frac{1}{{\bf m}(x)} + O(1)$, is the special case $B_{r,t} = \{s_t\}$ of the LDS coding theorem.

Our next objective is to use the LDS coding theorem to obtain useful bounds on the number of times that the value $K(S)$ is attained or approximated.

\begin{definition}
Let $S \subseteq\{0,1\}^*$ and $d \in \mathbb{N}$.
\begin{enumerate}
\item A \emph{d-approximate K-minimizer} of $S$ is a string $x \in S$ for which $K(x) \leq K(S) + d$.
\item A \emph{K-minimizer} of $S$ is a 0-approximate $K$-minimizer of $S$.
\end{enumerate}
\end{definition}
We use the LDS coding theorem to prove the following.

\begin{theorem}\label{lds block bound}
For every computably enumerable layered disjoint system $B$ there is a constant $c_B \in \mathbb{N}$ such that, for all $r,t,d \in \mathbb{N}$, the block $B_{r,t}$ has at most $2^{d+K(r)+c_B}$ $d$-approximate $K$-minimizers.
\end{theorem}

\begin{proof}
Let $B$ be a computably enumerable LDS, and let $c_B$ be as in the LDS coding theorem. Let $r,t,d \in \mathbb{N}$, and let $N$ be the number of $d$-approximate $K$-minimizers of the block $B_{r,t}$. Then
\begin{displaymath}
{\bf m}(B_{r,t}) \geq N \cdot 2^{-(K(B_{r,t}) + d)},
\end{displaymath}
so the LDS coding theorem tells us that
\begin{align*}
K(B_{r,t}) &\leq \log \frac{1}{N \cdot 2^{-(K(B_{r,t}) + d)}} + K(r) + c_B\\
					 &= K(B_{r,t}) + d - \log N + K(r) + c_B.
\end{align*}
This implies that
\begin{displaymath}
\log N \leq d + K(r) + c_B,
\end{displaymath}
whence
\begin{displaymath}
N \leq 2^{d + K(r) + c_B}. \qedhere
\end{displaymath}
\end{proof}

We now lift our terminology and notation to Euclidean space $\mathbb{R}^n$. In this context, a layered disjoint system is a set $B \subseteq \mathbb{N} \times \mathbb{N} \times \mathbb{R}^n$ such that, for all $r,s,t \in \mathbb{N}$,
\begin{displaymath}
s \neq t \Rightarrow B_{r,s} \cap B_{r,t} = \emptyset.
\end{displaymath}
We lift our Kolmogorov complexity notation and terminology to $\mathbb{R}^n$ in two steps:
\begin{enumerate}
\item Lifting to $\mathbb{Q}^n$: Each rational point $q \in \mathbb{Q}^n$ is encoded as a string $x \in \{0,1\}^*$ in a natural way. We then write $K(q)$ for $K(x)$. In this manner, $K(S),\,{\bf m}(S),\, K$-minimizers, and $d$-approximate $K$-minimizers are all defined for sets $S \subseteq \mathbb{Q}^n$.
\item Lifting to $\mathbb{R}^n$. For $S \subseteq \mathbb{R}^n$, we define $K(S) = K(S \cap \mathbb{Q}^n)$ and ${\bf m}(S) = {\bf m}(S \cap \mathbb{Q}^n)$. Similarly, a $K$-minimizer for $S$ is a $K$-minimizer for $S \cap \mathbb{Q}^n$, etc.
\end{enumerate}

For each $r \in \mathbb{N}$ and each $m = (m_1, \ldots, m_n) \in \mathbb{Z}^n$, let 
\begin{displaymath}
Q^{(r)}_m = [m_1\cdot2^{-r},\,(m_1+1)\cdot2^{-r})\times\cdots\times[m_n\cdot2^{-r},\,(m_n+1)\cdot2^{-r})
\end{displaymath}
be the \emph{r-dyadic cube} at $m$. Note that each $Q^{(r)}_m$ is ``half-open, half-closed'' in such a way that, for each $r \in \mathbb{N}$, the family
\begin{displaymath}
\mathcal{Q}^{(r)} = \{Q^{(r)}_m\,|\,m\in\mathbb{Z}^n\}
\end{displaymath}
is a partition of $\mathbb{R}^n$. It follows that (modulo trivial encoding) the collection
\begin{displaymath}
\mathcal{Q} = \{Q^{(r)}_m\,|\,r \in \mathbb{N}\textrm{ and } m \in \mathbb{Z}^n\}
\end{displaymath}
of all \emph{dyadic cubes} is a layered disjoint system whose $r$th layer is $\mathcal{Q}^{(r)}$. Moreover, the set
\begin{displaymath}
\{(r,m,q) \in \mathbb{N} \times \mathbb{Z}^n \times \mathbb{Q}^n\,|\,q \in Q^{(r)}_m\}
\end{displaymath}
is decidable, so Theorem \ref{lds block bound} has the following useful consequence.

\begin{corollary}\label{rdyadic cube bound}
There is a constant $c \in \mathbb{N}$ such that, for all $r,d \in \mathbb{N}$, no $r$-dyadic cube has more than $2^{d+K(r)+c}$ $d$-approximate $K$-minimizers. In particular, no $r$-dyadic cube has more than $2^{K(r)+c}$ $K$-minimizers.
\end{corollary}

The Kolmogorov complexity of an arbitrary point in Euclidean space depends on both the point and a precision parameter.

\begin{definition}
Let $x \in \mathbb{R}^n$ and $r \in \mathbb{N}$. The \emph{Kolmogorov complexity} of $x$ at \emph{precision} $r$ is
\begin{displaymath}
K_r(x) = K(B_{2^{-r}}(x)).
\end{displaymath}
\end{definition}

That is, $K_r(x)$ is the number of bits required to specify \emph{some} rational point in the open ball $B_{2^{-r}}(x)$. Note that, for each $q \in \mathbb{Q}^n$, $K_r(q) \nearrow K(q)$ as $r \rightarrow \infty$.

Given an open ball $B$ of radius $\rho$ and a real number $\alpha > 0$, we write $\alpha B$ for the ball with the same center as $B$ and radius $\alpha\rho$. We also write $\overline{B}$ for the topological closure of $B$.

The definition of $K_r(x)$ directs our attention to the Kolmgorov complexities of arbitrary balls of radius $2^{-r}$ in Euclidean space. The following easy fact is repeatedly useful in this context.

\begin{observation}\label{openball observe}
For every open ball $B \subseteq \mathbb{R}^n$ of radius $2^{-r}$,
\begin{displaymath}
B \cap 2^{-(r + \lceil \frac{1}{2}\log\,n\rceil)}\mathbb{Z}^n \neq \emptyset.
\end{displaymath}
\end{observation}

\begin{proof}
If $B$ is such a ball, then the expanded ball
\begin{displaymath}
B' = 2^{r + \lceil \frac{1}{2} \log\,n \rceil}B
\end{displaymath}
has radius
\begin{displaymath}
2^{\lceil \frac{1}{2} \log\,n \rceil} > 2^{\frac{1}{2} \log\,n - 1} = \frac{\sqrt{n}}{2}.
\end{displaymath}
This implies that
\begin{displaymath}
B' \cap \mathbb{Z}^n \neq \emptyset,
\end{displaymath}
whence
\begin{align*}
B \cap 2^{-(r + \lceil \frac{1}{2} \log\,n \rceil)}\mathbb{Z}^n &= 2^{-(r + \lceil \frac{1}{2} \log\,n \rceil)}(B' \cap \mathbb{Z}^n)\\
&\neq \emptyset. \qedhere
\end{align*}
\end{proof}

We use Observation \ref{openball observe} to establish the following connection between the complexities of cubes and the complexities of balls.

\begin{lemma}\label{ballcube lemma}
There is a constant $c \in \mathbb{N}$ such that, for every $r \in \mathbb{N}$, every $r$-dyadic cube $Q$, and every open ball $B \subseteq \mathbb{R}^n$ of radius $2^{-r}$ that intersects $Q$,
\begin{displaymath}
K(B) \leq K(Q) + K(r) + c.
\end{displaymath}
\end{lemma}

\begin{proof}
Fix a computable enumeration $m_0,m_1,m_2,\cdots$ of $\mathbb{Z}^n$ satisfying $|m_i| \leq |m_{i+1}|$ for all $i \in \mathbb{N}$. Note that, for all $i \in \mathbb{N}$,
\begin{equation}\label{3ib}
i < |\overline{B_{|m_i|}(0)} \cap \mathbb{Z}^n| \leq (2|m_i| + 1)^n.
\end{equation}
Let $l = \lceil \frac{1}{2}\log\,n \rceil$, and let $M$ be a self-delimiting Turing machine such that, if $U(\pi_1) = q \in \mathbb{Q}^n$ and $U(\pi_2) = r \in \mathbb{N}$, then, for all $i \in \mathbb{N}$,
\begin{equation}\label{3u}
M(\pi_1\pi_20^{|s_i|}1s_i) = q + 2^{-(r+l)}m_i.
\end{equation}
Let $c = 2\lceil 2n\log\,(1 + \sqrt{n}) \rceil + 1 + c_M$, where $c_M$ is an optimality constant for $M$.

Now assume the hypothesis, and let $q$ be a $K$-minimizer of $Q$. Observation \ref{openball observe} tells us that there is a point $m \in \mathbb{Z}^n$ such that $2^{-(r+l)}m \in B - q$. Then $|2^{-(r+l)}m|$ is the distance from a point in $B$ to the point $q \in \mathbb{Q}$, so
\begin{displaymath}
|m| = 2^{r+l}|2^{-(r+l)}m| \leq 2^{r+l}diam(B \cup Q).
\end{displaymath}
Since $B \cap Q \neq \emptyset$, it follows that
\begin{align}
|m| &\leq 2^{r+l}[diam(B) + diam(Q)]\notag\\
		&= 2^l(2 + \sqrt{n})\label{3mc}\\
		&\leq \frac{\sqrt{n}}{2}(2 + \sqrt{n})\notag\\
		&= \frac{n}{2} + \sqrt{n}\notag.
\end{align}
It is crucial here that this bound does not depend on $B$, $Q$, or $r$.

Choose $i \in \mathbb{N}$ such that $m_i = m$. By (\ref{3ib}) and (\ref{3mc}),
\begin{align}\label{3i2}
i < (2(\frac{n}{2} + \sqrt{n}) + 1)^n = (1 + \sqrt{n})^{2n}.
\end{align}
Now let $\pi = \pi_1\pi_20^{|s_i|}1s_i$, where $\pi_1$ and $\pi_2$ are minimum-length programs for $q$ and $r$, respectively. By (\ref{3u}) we have
\begin{displaymath}
M(\pi) = q + 2^{-(r+l)}m_i \in B.
\end{displaymath}
It follows by (\ref{3i2}) that
\begin{align*}
K(B) &\leq K(q + 2^{-(r+l)}m_i)\\
		 &\leq K_M(q + 2^{-(r+l)}m_i) + c_M\\
		 &\leq |\pi| + c_M\\
		 &= K(q) + K(r) + 2|s_i| + 1 + c_M\\
		 &= K(Q) + K(r) + 2\lceil 2n\log\,(1 + \sqrt{n}) \rceil + 1 + c_M\\
		 &= K(Q) + K(r) + c. \qedhere
\end{align*}
\end{proof}

\begin{theorem}\label{openball bound}
There is a constant $c \in \mathbb{N}$ such that, for all $r,d \in \mathbb{N}$, no open ball of radius $2^{-r}$ has more than $2^{d+2K(r)+c}$ $d$-approximate $K$-minimizers. In particular, no open ball of radius $2^{-r}$ has more than $2^{2K(r)+c}$ $K$-minimizers.
\end{theorem}

\begin{proof}
Let $B$ be an open ball of radius $2^{-r}$, let $Q$ be a $r$-dyadic cube such that $B \cap Q = \emptyset$, and let $u = K(B) - K(Q)$. There are at most $2^{d+u+K(r)+c'}$ $(d+u)$-approximate $K$-minimizers $q \in \mathbb{Q}$ of $Q$ such that $K(q) \leq K(Q) + d + u = K(B) + d$ where $c' \in \mathbb{N}$ is a constant from Corollary \ref{rdyadic cube bound}. Therefore, there are at most $2^{d+u+K(r)+c'}$ $d$-approximate $K$-minimizers of $B$ in $Q \cap B$.

Observe that it takes at most $3^n = 2^{n\log\,3}$ $r$-dyadic cubes to cover $B$. By Lemma \ref{ballcube lemma}, $u \leq K(r) + c''$, where $c'' \in \mathbb{N}$ is a constant. Therefore, it follows that $B$ has at most $2^{d+2K(r)+c}$ $d$-approximate $K$-minimizers where $c = c' + c'' + n\,\log\,3$. In particular, $B$ has at most $2^{2K(r) + c}$ $K$-minimizers. \qedhere
\end{proof}

Lemma 3.5 also gives a slightly simplified proof of the known upper bound on $K_r(x)$.

\begin{observation}[\cite{jLutMay08}]\label{kr observe}
For all $x \in \mathbb{R}^n$, $K_r(x) \leq nr + o(r)$.
\end{observation}

\begin{proof}
Let $c$ be a constant of Lemma \ref{ballcube lemma}, let $x = (x_1, \ldots, x_n) \in \mathbb{R}^n$, and let
\begin{displaymath}
\gamma_x = \max\{|x_i| + 1\,\big| 1 \leq i \leq n\}.
\end{displaymath}
For each $r \in \mathbb{N}$, let $m(r) = (m_1, \ldots, m_n)$ be the unique $m \in \mathbb{Z}^n$ such that $x \in \mathbb{Q}^{(r)}_{m}$. Then, for each $r \in \mathbb{N}$ and $1 \leq i \leq n$, we have $|m_i| \leq 2^r\gamma_x$. It follows easily from this that there is a constant $c' \in \mathbb{N}$ such that, for every $r \in \mathbb{N}$,
\begin{equation}\label{kofm}
K(m(r)) \leq n(\log(2^r\gamma_x) + 2\log\log(2^r\gamma_x)) + c_1.
\end{equation}
There is clearly a constant $c_2 \in \mathbb{N}$ such that, for every $r \in \mathbb{N}$,
\begin{equation}\label{kof2m}
K(2^{-r}m(r)) \leq K(m(r)) + K(r) + c_2.
\end{equation}
By (\ref{kofm}), (\ref{kof2m}), and Lemma \ref{ballcube lemma} we now have
\begin{align*}
K_r(x) &= K(B_{2^{-r}}(x))\\
			 &\leq K(Q^{(r)}_{m(r)}) + K(r) + c\\
			 &\leq K(m(r)) + K(r) + c\\
			 &\leq nr + \epsilon(r),
\end{align*}
where
\begin{align*}
\epsilon(r) &= n(\log \gamma_x + 2\log\log(2^r\gamma_x)) + 2K(r) + c + c_1 + c_2.\\
						&= o(r)
\end{align*}
as $r \rightarrow \infty$.
\end{proof}

\begin{lemma} \label{krs lemma}
There is a constant $c \in \mathbb{N}$ such that, for all $r,s \in \mathbb{N}$, $x \in \mathbb{R}^n$, and $q \in B_{2^{-r}}(x)$,
\begin{align*}
K_{r+s}(x) \leq K(q) + ns + K(r) + a_s,
\end{align*}
where $a_s = K(s) + 2\log(\lceil \frac{1}{2}\log{n} \rceil + s + 3) + n(\lceil \frac{1}{2}\log{n} \rceil + 3) + K(n) + 2\log{n} + c$.
\end{lemma}

\begin{proof}
Fix a computable enumeration $m_0, m_1, m_2, \cdots$ of $\mathbb{Z}^n$ satisfying $|m_i| \leq |m_{i+1}|$ for all $i \in \mathbb{N}$. Note that, for all $i \in \mathbb{N}$,
\begin{equation}\label{3i}
i < |\overline{B_{|m_i|}(0)} \cap \mathbb{Z}^n| \leq (2|m_i| + 1)^n.
\end{equation}

Let $l = \lceil \frac{1}{2} \log{n} \rceil$, and let $M$ be a self-delimiting Turing machine such that, if $U(\pi_1) = q \in \mathbb{Q}^n$, $U(\pi_2) = r \in \mathbb{N}$, $U(\pi_3) = s \in \mathbb{N}$, $U(\pi_4) = n \in \mathbb{N}$, and $U(\pi_5) = i \in \mathbb{N}$, then
\begin{equation}\label{3m}
M(\pi_1\pi_2\pi_3\pi_4\pi_5) = q + 2^{-(r+s+l+1)}m_i.
\end{equation}
Let $a_s = 2n(\lceil \frac{1}{2} \log{n} \rceil + s + 3) + 1 + c_M$, where $c_M$ is an optimality constant for $M$.

Now assume the hypothesis. Observation \ref{openball observe} tells us that there is a point $m \in \mathbb{Z}^n$ such that $2^{-(r+s+l)}m \in B_{2^{-(r+s)}}(x) - q$. Then $|2^{-(r+s+l)}m|$ is the distance from a point in $B_{2^{-(r+s)}}(x)$ to the point $q$, so
\begin{align}
|m| &= 2^{r+s+l}|2^{-(r+s+l)}m|\notag\\
		&\leq 2^{r+s+l}(2^{-r} + 2^{-(r+s)})\notag\\
		&= 2^{s+l}(1 + 2^{-s})\label{3ma}\\
		&=2^l(2^s + 1)\notag\\
		&\leq 2^l2^{s + 1}\notag\\
		&\leq 2^{l + s + 1}\notag.
\end{align}
Choose $i \in \mathbb{N}$ such that $m_i = m$. By (\ref{3i}) and (\ref{3ma}),
\begin{equation}\label{3ia}
i < (2|m_i| + 1)^n \leq (2(2^{l+s+1}) + 1)^n = (2^{l+s+2} + 1)^n.
\end{equation}

Now let $\pi = \pi_1\pi_2\pi_3\pi_4\pi_5$, where $\pi_1$, $\pi_2$, $\pi_3$, $\pi_4$, and $\pi_5$ are minimum-length programs for $q$, $r$, $s$, $n$, and $i$, respectively. By (\ref{3m}) we have
\begin{equation}\label{3mpi}
M(\pi) = q + 2^{-(r+s+l+1)}m_i \in B_{2^{-(r+s)}}(x).
\end{equation}
Therefore, (\ref{3mpi}) and optimality tell us that
\begin{align*}
K_{r+s}(x) &= K(B_{2^{-(r+s)}}(x))\\
					 &\leq K(q + 2^{-(r+l)}m_i)\\
					 &\leq K_M(q + 2^{-(r+l)}m_i) + c_M\\
					 &= |\pi| + c_M\\
					 &= K(q) + K(r) + K(s) + K(n) + K(i) +c_M.
\end{align*}
As noted in section 2, there is a constant $c_0 \in \mathbb{N}$ such that
\begin{align*}
K(i) \leq \log(1+i) + 2\log(1 + \log(1+i)) + c_0.
\end{align*}
It follows by (\ref{3ia}) that
\begin{align*}
K(i) &\leq n\log(1+2^{l+s+2}) + 2\log(1+n\log(1+2^{l+s+2})) + c_0\\
		 &\leq n(l+s+3) + 2\log(1+n(l+s+3)) + c_0\\
		 &\leq n(l+s+3) + 2(1 + \log{n} + \log(l+s+3)) + c_0\\
		 &= ns + n(l+3) + 2\log{n} + 2\log(l+s+3) + c_0 + 2.
\end{align*}
Letting $c = c_M + c_0 + 2$, it follows that
\begin{align*}
K_{r+s}(x) \leq K(q) + ns + a_s,
\end{align*}
where $a_s = K(s) + 2\log(l+s+3) + n(l+3) + K(n) + 2\log{n} + c$. \qedhere
\end{proof}

The following corollary says roughly that, in $\mathbb{R}^n$, precision can be improved by $ns$ bits by adding $ns$ bits of specification.

\begin{corollary}
There is a constant $c \in \mathbb{N}$ such that, for all $r,s \in \mathbb{N}$ and $x \in \mathbb{R}^n$,
\begin{displaymath}
K_{r+s}(x) \leq K_r(x) + ns + b_s,
\end{displaymath}
where $b_s = a_s + K(r)$ and $a_s$ is as in Lemma \ref{krs lemma}.
\end{corollary}
\section{Algorithmic Mutual Information in Euclidean Space}

\newcommand{\Q}{\mathbb{Q}}

This section develops the algorithmic mutual information between points in Euclidean space at a given precision.  As in section 3, we assume that rational points $q \in \Q^n$ are encoded as binary strings in some natural way.  Mutual information between rational points is then defined from conditional Kolmogorov complexity in the standard way \cite{bLiVit08} as follows.
\begin{definition}
Let $p \in \Q^m$, $r \in \Q^n$, $s \in \Q^t$.
\begin{enumerate}
\item The {\it mutual information} between $p$ and $q$ is
\[
I(p:q) = K(q)-K(q\,|\,p).
\]
\item The {\it mutual information} between $p$ and $q$ {\it given} $s$ is
    \[
    I(p:q\,|\,s)=K(q\,|\,s)-K(q\,|\,p,s).
    \]
\end{enumerate}
\end{definition}

The following properties of mutual information are well known \cite{bLiVit08}.

\begin{theorem}\label{mutual info properties}
Let $p \in \Q^m$ and $q \in \Q^n$.
\begin{enumerate}
\item $I(p, K(p) :q) = K(p) + K(q) - K(p,q)+O(1).$
\item $I(p,K(p):q) = I(q, K(q):p)+O(1).$
\item $I(p:q) \leq \min\left\{K(p), K(q) \right\} + O(1)$.
\end{enumerate}
\end{theorem}
\vspace*{12pt}
(Each of the properties 1 and 2 above is sometimes called {\it symmetry of mutual information}.)

Mutual information between points in Euclidean space at a given precision is now defined as follows.

\begin{definition}
The {\it mutual information} of $x \in \mathbb{R}^n$ and $y \in \mathbb{R}^t$ at {\it precision} $r \in \mathbb{N}$ is
\begin{displaymath}
I_r(x:y) = \min\{I(q_x:q_y)\,|\,q_x \in B_{2^{-r}}(x) \cap \mathbb{Q}^n \textrm{ and } q_y \in B_{2^{-r}}(y) \cap \mathbb{Q}^t\}.
\end{displaymath}
\end{definition}

As noted in the introduction, the role of the minimum in the above definition is to eliminate ``spurious'' information that points $q_x \in B_{2^{-r}} \cap \Q^n$ and $q_y \in B_{2^{-r}}(y) \cap \Q^t$ might share for reasons not forced by their proximities to $x$ and $y$, respectively.

\begin{notation}
We also use the quantity
\[
J_r(x:y) = \min\{I(q_x:q_y)\,|\,p_x \textnormal{ is a K-minimizer of } B_{2^{-r}}(x) \textnormal{ and } p_y \textnormal{ is a K--minimizer of}\, B_{2^{-r}}(y)\}.
\]
\end{notation}

Although $J_r(x:y)$, having two ``layers of minimization'', is somewhat more involved than $I_r(x : y)$, one can imagine using it as the definition of mutual information.
In fact, for all $x,y \in \mathbb{R}$, $J_r(x:y)$ does not differ greatly from $I_r(x:y)$. We next develop machinery for proving this useful fact, which is Theorem \ref{ij theorem} below.

\begin{lemma}\label{set bound}
There is a constant $c \in \mathbb{N}$ such that, for any $r \in \mathbb{N}$, open ball $B \subseteq \mathbb{R}^n$ of radius $2^{-r}$, and $q \in B \cap \mathbb{Q}^n$,
\begin{align*}
|\{p' \in B_{2^{1-r}}(q) \cap \mathbb{Q}^n\,|\,K(p') \leq K(B)\}| \leq 2^{K(r) + 2K(r-1) + c}.
\end{align*}
\end{lemma}

\begin{proof}
Let $B$ be centered at $x \in \mathbb{R}^n$. If $p_q \in \mathbb{Q}^n$ is a $K$-minimizer of $B_{2^{1-r}}(q)$, then $p_q \in B_{2^{2-r}}(x)$. By Lemma \ref{krs lemma},
\begin{align*}
K(B) &\leq K(p_q) + K(r) + c\\
		 &= K(B_{2^{1-r}}(q)) + K(r) + c,
\end{align*}
where $c = K(2) + K(n) + 2n(\lceil \frac{1}{2} \log\,n \rceil + 5) + 1 + c'$ for some constant $c'$. This inequality implies that any $K$-minimizer of $B$ is also a $K(r) + c$-approximate $K$-minimizer of $B_{2^{1-r}}(q)$. Therefore, by Lemma \ref{openball bound},
\begin{align*}
|\{p' \in B_{2^{1-r}}(q) \cap \mathbb{Q}^n\,|\,K(p') \leq K(B)\}| &\leq |\{p' \in B_{2^{1-r}}(q) \cap \mathbb{Q}^n\,|\,K(p') \leq K(B_{2^{1-r}}(q)) + K(r) + c\}|\\
																																	&\leq 2^{K(r) + 2K(r-1) + c}. \qedhere
\end{align*}
\end{proof}

\begin{lemma}\label{arb rat}
For all $x \in \mathbb{R}^n$, $q \in \mathbb{Q}^t$, and $q_x,p_x \in B_{2^{-r}}(x) \cap \mathbb{Q}^n$ where $p_x$ is a $K$-minimizer of $B_{2^{-r}}(x)$,
\begin{align*}
K(q\,|\,q_x) \leq K(q\,|\,p_x) + K(K(p_x)) + o(r).
\end{align*}
\end{lemma}

\begin{proof}
Let $M$ be a self-delimiting Turing machine that takes programs of the form $\pi = \langle \pi_1 \pi_2 \pi_3 0^{|s_i|}1s_i,q \rangle$, where $U(\pi_1,p) = q' \in \mathbb{Q}^t$, $U(\pi_2) = K(p)$, $U(\pi_3) = r \in \mathbb{N}$, and $i \in \mathbb{N}$. $M$ runs $\pi_2$ and $\pi_3$ on $U$ to obtain $K(p)$ and $r$, performs a systematic search for the $i^{th}$ discovered element of $\{p' \in B_{2^{1-r}}(q) \cap \mathbb{Q}^n\,|\,K(p') \leq K(p)\}$, and outputs $U(\langle \pi_1, p_i \rangle)$. Therefore,
\begin{equation}\label{4m}
M(\pi) = U(\langle \pi_1, p_i \rangle).
\end{equation}
Let $c_M$ be an optimality constant for $M$.

Assume the hypothesis, and let $\pi = \langle \pi_1\pi_2\pi_30^{|s_i|}1s_i, q_x \rangle$, where $\pi_1$ is a minimum-length program for $q$ when given $p_x$, $\pi_2$ is a minimum-length program for $K(p_x)$, $\pi_3$ is a minimum-length program for $r$, and $i$ is an index for $p_x$ in the set $\{p' \in B_{2^{1-r}}(q_x) \cap \mathbb{Q}^n\,|\,K(p') \leq K(p_x)\}$. By (\ref{4m}), we have $M(\pi) = U(\langle \pi_1, p_x \rangle) = q$. Therefore, by Lemma \ref{set bound} and optimality,
\begin{align*}
K(q\,|\,q_x) &\leq K_M(q\,|\,q_x) + c_M\\
						 &\leq |\pi_1\pi_2\pi_30^{|s_i|}1s_i| + c_M\\
						 &= K(q\,|\,p_x) + K(K(p_x)) + K(r) + 2|s_i| + 1 + c_M\\
						 &\leq K(q\,|\,p_x) + K(K(p_x)) + K(r) + 2\log{|\{p' \in B_{2^{1-r}}(q_x) \cap \mathbb{Q}^n\,|\,K(p') \leq K(p_x)\}|} + 1 + c_M\\
						 &\leq K(q\,|\,p_x) + K(K(p_x)) + K(r) + 2(K(r) + 2K(r-1) + c) + 1 + c_M\\
						 &= K(q\,|\,p_x) + K(K(p_x)) + o(r). \qedhere
\end{align*}
\end{proof}

By Lemma \ref{arb rat} and Observation \ref{kr observe} we have the following.

\begin{corollary}\label{kmin given rat}
Let $x \in \mathbb{R}^n$. If $q_x \in B_{2^{-r}}(x) \cap \mathbb{Q}^n$ and $p_x \in \mathbb{Q}^n$ is a $K$-minimizer of $B_{2^{-r}}(x)$, then $K(p_x\,|\,q_x) = o(r)$.
\end{corollary}

\begin{lemma}\label{kmin given kmin}
Let $x \in \mathbb{R}^n$ and $y \in \mathbb{R}^t$. If $p_x \in B_{2^{-r}}(x)$ and $q_y,p_y \in B_{2^{-r}}(y)$ where $p_x$ is a $K$-minimizer for $B_{2^{-r}}(x)$ and $p_y$ is a $K$-minimizer for $B_{2^{-r}}(y)$, then
\begin{align*}
K(p_x\,|\,q_y,K(q_y)) \leq K(p_x\,|\,p_y,K(p_y)) + o(r).
\end{align*}
\end{lemma}

\begin{proof}
By the triangle inequality for strings and Corollary \ref{kmin given rat},
\begin{align*}
K(p_x\,|\,q_y,K(q_y)) &\leq K(p_x\,|\,p_y,K(p_y)) + K(p_y\,|\,q_y,K(q_y)) + O(1)\\
											&\leq K(p_x\,|\,p_y,K(p_y)) + K(p_y\,|\,q_y) + O(1)\\
											&= K(p_x\,|\,p_y,K(p_y)) + o(r). \qedhere
\end{align*}
\end{proof}
The following lemma was inspired by Hammer et al. \cite{jHRSV00}.
\begin{lemma}
For all $x,y,z \in \{0,1\}^*$,
\begin{align*}
K(z) - K(K(z)) - K(K(x)) &\leq I(x:y) + K(z\,|\,x,K(x)) + K(z\,|\,y,K(y)) - K(z\,|\,\langle x,y \rangle,K(\langle x,y \rangle))\\
													 & \hspace*{5mm} - I(x:y|z) + O(1).
\end{align*}
\end{lemma}
\begin{proof}
By the well-known identity (\ref{kxy}), obvious inequalities, and basic definitions.
\begin{align*}
&\hspace*{5mm}K(z) - K(K(z)) - K(K(x))\\
&= K(x) - K(x,y) - K(K(x)) + K(x,z) - K(x) + K(y,z) - K(x,y,z)\\
&\hspace*{5mm}+ K(x,y) + K(z) - K(z,y) - K(K(z)) + K(x,z,y) - K(x,z) + O(1)\\
&= - K(y\,|\,x,K(x)) - K(K(x)) + K(x,z) - K(x) + K(y,z) - K(x,y,z)\\
&\hspace*{5mm}+ K(x,y) - K(y\,|\,z,K(z)) - K(K(z)) + K(y\,|\,x,z, K(x,z)) + O(1)\\
&\leq K(y) - K(y\,|\,x) + K(x,z) - K(x) + K(y,z) - K(y) - K(x,y,z) + K(x,y)\\
&\hspace*{5mm}- K(y\,|\,z) + K(y\,|\,x,z) + O(1)\\
&= I(x:y) + K(z\,|\,x,K(x)) + K(z\,|\,y,K(y)) - K(z\,|\,x,y, K(x,y)) - I(x:y\,|\,z) + O(1). \qedhere
\end{align*}
\end{proof}

\begin{corollary}\label{long ineq}
For all $x,y,z \in \{0,1\}^*$, 
\begin{equation*}
I(x:y) \geq K(z) - K(z\,|\,x,K(x)) - K(z\,|\,y,K(y)) - K(K(x)) - K(K(z)) + O(1).
\end{equation*}
\end{corollary}

\begin{theorem}\label{ij theorem}
For all $x \in \mathbb{R}^n$ and $y \in \mathbb{R}^t$,
\begin{align*}
I_r(x:y) = J_r(x:y) + o(r).
\end{align*}
\end{theorem}

\begin{proof}
Let $q_x,p_x \in \mathbb{Q}^n$ and $q_y,p_y \in \mathbb{Q}^t$ where $p_x$ is a $K$-minimizer of $B_{2^{-r}}(x)$, $p_y$ is a $K$-minimizer of $B_{2^{-r}}(y)$, and $I(q_x:q_y) = I_r(x:y)$. By Lemma \ref{arb rat},
\begin{displaymath}
K(q_y) - K(q_y\,|\,p_x) \leq K(q_y) - K(q_y\,|\,q_x) + K(K(p_x)) + o(r).
\end{displaymath}
Applying the definition of mutual information for rationals, we have
\begin{displaymath}
I(p_x:q_y) \leq I(q_x:q_y) +K(K(p_x)) + o(r),
\end{displaymath}
which, by Corollary \ref{long ineq} and Observation \ref{kr observe}, implies that
\begin{align*}
I(q_x:q_y) &\geq K(p_x) - K(p_x\,|\,p_x,K(p_x)) - K(p_x\,|\,q_y,K(q_y)) + o(r)\\
					 &= K(p_x) - K(p_x\,|\,q_y,K(q_y)) + o(r).
\end{align*}
By applying Lemma \ref{kmin given kmin} and the definition of mutual information for rationals to the above inequality, we obtain
\begin{align*}
I(q_x:q_y) &\geq K(p_x) - K(p_x\,|\,p_y,K(p_y)) + o(r)\\
					 &= I(p_y,K(p_y):p_x) + o(r).
\end{align*}
Thus, by Theorem \ref{mutual info properties},
\begin{align*}
I(q_x:q_y) &\geq I(p_x,K(p_x):p_y) + o(r)\\
					 &\geq I(p_x:p_y) + o(r).
\end{align*}

The above inequality tells us that $I_r(x:y) = I(q_x:q_y) \geq I(p_x:p_y) + o(r) = J_r(x:y) + o(r)$. Also, by definition, $I_r(x:y) \leq J_r(x:y)$. \qedhere
\end{proof}

Before discussing the properties of $I_r(x:y)$, we need one more lemma.

\begin{lemma}\label{k lessthan kr}
Let $x \in \mathbb{R}^n$, $y \in \mathbb{R}^t$, and $r \in \mathbb{N}$. If $p_x \in \mathbb{Q}^n$ is a $K$-minimizer of $B_{2^{-r}}(x)$ and $p_y \in \mathbb{Q}^t$ is a $K$-minimizer of $B_{2^{-r}}(y)$, then
\begin{align*}
K(p_x,p_y) = K_r(x,y) + o(r).
\end{align*}
\end{lemma}

\begin{proof}
By Corollary \ref{kmin given rat},
\begin{align*}
K_r(x,y) \leq K(p_x,p_y) &\leq K(p_y) + K(p_x\,|\,p_y)\\
					 &= K_r(y) + K(p_x\,|\,p_y)\\
					 &\leq K_r(x,y) + K(p_x\,|\,p_y) + O(1)\\
					 &= K_r(x,y) + o(r). \qedhere
\end{align*}
\end{proof}
The following characterization of algorithmic (Martin-L\"{o}f) randomness is well known.
\begin{definition}
A point $x \in \mathbb{R}^n$ is \emph{random} if there is a constant $d \in \mathbb{N}$ such that, for all $r \in \mathbb{N}$,
\begin{equation*}
K_r(x) \geq nr - d.
\end{equation*}
Two points $x \in \mathbb{R}^n$ and $y \in \mathbb{R}^t$ are \emph{independently random} if the point $(x,y) \in \mathbb{R}^{n+t}$ is random.
\end{definition}

We now establish the following useful properties of $I_r(x:y)$.

\begin{theorem}\label{real mutual info properties}
For all $x \in \mathbb{R}^n$ and $y \in \mathbb{R}^t$,
\begin{enumerate}
\item $I_r(x:y) = K_r(x) + K_r(y) - K_r(x,y) + o(r)$.
\item $I_r(x:y) \leq \min\{K_r(x),K_r(y)\} + o(r)$.
\item \text{If } $x$ \text{ and } $y$ \text{ are independently random, then } $I_r(x:y) = o(r)$.
\item $I_r(x:y) = I_r(y:x) + o(r)$.
\end{enumerate}
\end{theorem}

\begin{proof}
To prove the first statement, let $p_x \in \mathbb{Q}^n$ be a $K$-minimizer of $B_{2^{-r}}(x)$ and $p_y \in \mathbb{Q}^t$ be a $K$-minimizer of $B_{2^{-r}}(y)$. First, by Theorem \ref{ij theorem},
\begin{align*}
I_r(x:y) &= J_r(x:y) + o(r)\\
	       &= I(p_x:p_y) + o(r)\\
	       &= K(p_y) - K(p_y\,|\,p_x) + o(r)\\
	       &\leq K(p_y) - K(p_y\,|\,p_x,K(p_x)) + o(r).
\end{align*}
By (\ref{kxy}) and Lemma \ref{k lessthan kr}, this implies that
\begin{align*}
I_r(x:y) &\leq K(p_y) + K(p_x) - K(p_x,p_y) + o(r)\\
	       &= K_r(y) + K_r(x) - K_r(x,y) + o(r).
\end{align*}
Next we show that $I_r(x:y) \geq K_r(x) + K_r(y) - K_r(x,y) + o(r)$. By the above inequality,
\begin{align*}
I_r(x:y) &= K(p_y) - K(p_y\,|\,p_x) + o(r)\\
		 		 &\geq K(p_y) - K(p_y\,|\,p_x,K(p_x)) - K(K(p_x)) + o(r).
\end{align*}
Finally, by (\ref{kxy}), Observation \ref{kr observe}, and Lemma \ref{k lessthan kr},
\begin{align*}
I_r(x:y) &\geq K(p_y) + K(p_x) - K(p_x,p_y) + o(r)\\
		 		 &\geq K_r(y) + K_r(x) - K_r(x,y) + o(r).
\end{align*}

We continue to the second statement. By 1,
\begin{align*}
I_r(x:y) &= K_r(x) + K_r(y) - K_r(x,y) + o(r)\\
				 &\leq K_r(x) + K_r(y) - K_r(y) + o(r)\\
				 &= K_r(x) + o(r).
\end{align*}
Likewise, $I_r(x:y) \leq K_r(y) + o(r)$. Therefore, $I_r(x:y) \leq \min\{K_r(x),K_r(y)\} + o(r)$.

We now prove the third statement. By 1,
\begin{align*}
I_r(x:y) &= K_r(x) + K_r(y) - K_r(x,y) + o(r)\\
		 &\leq K_r(x) + K_r(y) + K(r) - K_r(r,x,y) + o(r)\\
		 &\leq nr + tr + K(r) - (n+t)r + o(r)\\
		 &= o(r),
\end{align*}
where the last inequality is due to the premise that $x$ and $y$ are independently random and Observation \ref{kr observe}.

Lastly, we prove the fourth statement. By 1 and Lemma \ref{k lessthan kr},
\begin{align*}
I_r(x:y) &= K_r(x) + K_r(y) - K_r(x,y) + o(r)\\
				 &= K_r(x) + K_r(y) - K(p_x,p_y) + o(r)\\
				 &= K_r(x) + K_r(y) - K(p_y,p_x) + o(r)\\
				 &= K_r(x) + K_r(y) - K_r(y,x) + o(r)\\
				 &= I_r(y:x) + o(r). \qedhere
\end{align*}
\end{proof}
\section{Mutual Dimension in Euclidean Space}
We now define mutual dimensions between points in Euclidean space(s).
\begin{definition}
The \emph{lower} and \emph{upper mutual dimensions} between $x \in \mathbb{R}^n$ and $y \in \mathbb{R}^t$ are
\begin{displaymath}
mdim(x:y) = \displaystyle\liminf_{r \rightarrow \infty}\frac{I_r(x:y)}{r}
\end{displaymath}
and
\begin{displaymath}
Mdim(x:y) = \displaystyle\limsup_{r \rightarrow \infty}\frac{I_r(x:y)}{r},
\end{displaymath}
respectively.
\end{definition}

With the exception of the data processing inequality, which we prove in section 5, the following theorem says that the mutual dimensions $mdim$ and $Mdim$ have the basic properties that any mutual information measure should have. (See, for example, \cite{jBell62}.)

\begin{theorem}\label{mutual dim properties}
For all $x \in \mathbb{R}^n$ and $y \in \mathbb{R}^t$, the following hold.
\begin{enumerate}
\item $dim(x) + dim(y) - Dim(x,y) \leq mdim(x:y) \leq Dim(x) + Dim(y) - Dim(x,y)$.
\item $dim(x) + dim(y) - dim(x,y) \leq Mdim(x:y) \leq Dim(x) + Dim(y) - dim(x,y)$.
\item $mdim(x:y) \leq \min\{dim(x), dim(y)\},\,Mdim(x:y) \leq \min\{Dim(x), Dim(y)\}$.
\item $0 \leq mdim(x:y) \leq Mdim(x:y) \leq \min\{n,t\}$.
\item If $x$ and $y$ are independently random, then $Mdim(x:y) = 0$.
\item $mdim(x:y) = mdim(y:x),\,Mdim(x:y) = Mdim(y:x)$.
\end{enumerate}
\end{theorem}

\begin{proof}
To prove the first statement, we use Theorem \ref{real mutual info properties} and basic properties of $\limsup$ and $\liminf$. First we show that $mdim(x:y) \geq dim(x) + dim(y) - Dim(x,y)$.
\begin{align*}
mdim(x:y) &= \displaystyle\liminf_{r \rightarrow \infty}\frac{I_r(x:y)}{r} \\
					&= \displaystyle\liminf_{r \rightarrow \infty}\frac{K_r(x) + K_r(y) - K_r(x,y) + o(r)}{r}\\
					&\geq \displaystyle\liminf_{r \rightarrow \infty}\frac{K_r(x)}{r} +
					\displaystyle\liminf_{r \rightarrow \infty}\frac{K_r(y)}{r} +
					\displaystyle\liminf_{r \rightarrow \infty}\frac{-K_r(x,y)}{r} +
					\displaystyle\liminf_{r \rightarrow \infty}\frac{o(r)}{r}\\
					&= dim(x) + dim(y) - \displaystyle\limsup_{r \rightarrow \infty}\frac{K_r(x,y)}{r}\\
					&= dim(x) + dim(y) - Dim(x,y).
\end{align*}
Next we show that $mdim(x:y) \leq Dim(x) + Dim(y) - Dim(x,y)$.
\begin{align*}
mdim(x:y) &= Dim(x) + Dim(y) - Dim(x) - Dim(y) + mdim(x:y)\\
					&= Dim(x) + Dim(y) - \Big(\displaystyle\limsup_{r \rightarrow \infty}\frac{K_r(x)}{r} + \displaystyle\limsup_{r \rightarrow \infty}\frac{K_r(y)}{r} + \displaystyle\limsup_{r \rightarrow \infty}\frac{-I_r(x:y)}{r}\Big)\\
					&\leq Dim(x) + Dim(y) - \displaystyle\limsup_{r \rightarrow \infty}\frac{K_r(x) + K_r(y) - K_r(x) - K_r(y) + K_r(x,y) + o(r)}{r}\\
					&= Dim(x) + Dim(y) - Dim(x,y).
\end{align*}
The proof of the second statement is similar to the first. The third statement follows immediately from Theorem \ref{real mutual info properties} and the fact that $\displaystyle\liminf_{r \rightarrow \infty}\min\{K_r(x), K_r(y)\} \leq \min\{\displaystyle\liminf_{r \rightarrow \infty}K_r(x), \displaystyle\liminf_{r \rightarrow \infty}K_r(y)\}$. The fourth statement follows from the third and the fact that, for all $x \in \mathbb{R}^n$, $Dim(x) \leq n$. Finally, both the fifth and sixth statements follow immediately from Theorem \ref{real mutual info properties}. \qedhere
\end{proof}
\section{Data Processing Inequalities}
Our objectives in this section are to prove data processing inequalities for lower and upper mutual dimensions in Euclidean space.

The following result is the main theorem of this paper. The meaning and necessity of the Lipschitz hypothesis are explained in the introduction.

\begin{theorem}[data processing inequality]\label{dpi}
If $f:\mathbb{R}^n \rightarrow \mathbb{R}^t$ is computable and Lipschitz, then, for all $x \in \mathbb{R}^n$ and $y \in \mathbb{R}^t$,
\begin{displaymath}
mdim(f(x):y) \leq mdim(x:y)
\end{displaymath}
and
\begin{displaymath}
Mdim(f(x):y) \leq Mdim(x:y).
\end{displaymath}
\end{theorem}
We in fact prove a stronger result.

\begin{definition}
A \emph{modulus (of uniform continuity)} for a function $f: \mathbb{R}^n \rightarrow \mathbb{R}^k$ is a nondecreasing function $m: \mathbb{N} \rightarrow \mathbb{N}$ such that, for all $x, y \in \mathbb{R}^n$ and $r \in \mathbb{N}$,
\begin{align*}
|x - y| \leq 2^{-m(r)} \Rightarrow |f(x) - f(y)| \leq 2^{-r}.
\end{align*}
\end{definition}
Note that it is well known that a function is uniformly continuous if and only if it has a modulus of uniform continuity.
\begin{lemma}\label{kyx less than kyf}
For all strings $x,y,z \in \{0,1\}^*$ and all partial computable functions $f: \{0,1\}^* \times \{0,1\}^*$,
\begin{equation*}
K(y\,|\,x) \leq K(y\,|\,f(x,z)) + K(z) + O(1).
\end{equation*}
\end{lemma}

\begin{proof}
Let $M$ be a self-delimiting Turing machine such that if $U(\pi_1, f(x,z)) = y$, $U(\pi_2) = z$, and $\pi_3$ is a program for $f$ where $x,y,z \in \{0,1\}^*$ and $f: \{0,1\}^* \times \{0,1\}^*$ is a partial computable function, then
\begin{equation}\label{2m}
M(\pi_1\pi_2\pi_3, x) = y.
\end{equation}
Assume the hypothesis, and let $\pi = \pi_1\pi_2\pi_3$ where $\pi_1$ is a minimum-length program for $y$ given $f(x,z)$, $\pi_2$ is a minimum-length program for $z$, and $\pi_3$ is a minimum-length program for $f$. Therefore, by (\ref{2m}), we have $M(\pi,x) = y$. By optimality,
\begin{align*}
K(y\,|\,x) &\leq K_M(y\,|\,x) + c_M\\
					 &\leq |\pi| + c_M\\
					 &= K(y\,|\,f(x,z)) + K(z) + K(f) + c_M\\
					 &= K(y\,|\,f(x,z)) + K(z) + O(1). \qedhere
\end{align*}
\end{proof}

\begin{lemma}\label{mod dpi}
If $f:\mathbb{R}^n \rightarrow \mathbb{R}^k$ is computable and $m:\mathbb{N} \rightarrow \mathbb{N}$ is a computable modulus for $f$, then for every $x \in \mathbb{R}^n$, $y \in \mathbb{R}^t$,
\begin{align*}
I_r(f(x):y) \leq I_{m(r+1)}(x:y) + o(r).
\end{align*}
\end{lemma}

\begin{proof}
Let $q_x \in \mathbb{Q}^n$ and $q_y \in \mathbb{Q}^t$ such that $I_{m(r+1)}(x:y) = I(q_x:q_y)$. Because $|x - q_x| \leq 2^{-m(r+1)}$, where $m$ is a modulus for $f$, we know that $|f(x) - f(q_x)| \leq 2^{-(r+1)}$. Also, since $f$ is computable, there exists an oracle Turing machine $M$ that uses an oracle $q_x$ such that $|M^{q_x}(r) - f(q_x)| \leq 2^{-r}$.
Let $h:\mathbb{N} \times \mathbb{Q}^n \rightarrow \mathbb{Q}^k$ be a function such that $h(q_x, r) = M^{q_x}(r+1).$ Observe that
\begin{align*}
|h(q_x,r+1) - f(x)| &\leq |f(x) - f(q_x)| + |f(q_x) - h(q_x, r)|\\
										&\leq 2^{-(r+1)} + 2^{-(r+1)}\\
										&= 2^{-r}.
\end{align*}
From this and Lemma \ref{kyx less than kyf}, it follows that
\begin{align*}
I_r(f(x):y) &\leq I(M^{q_x}(r+1):q_y)\\
						&= I(h(q_x, r):q_y)\\
						&\leq I(q_x:q_y) + K(r) + O(1)\\
						&= I_{m(r+1)}(x:y) + o(r). \qedhere
\end{align*}
\end{proof}

\begin{lemma}[modulus processing lemma]\label{mod proc lemma}
If $f:\mathbb{R}^n \rightarrow \mathbb{R}^k$ is computable and $m$ is a computable modulus for $f$, then for all $x \in \mathbb{R}^n$ and $y \in \mathbb{R}^t$,
\begin{displaymath}
mdim(f(x):y) \leq mdim(x:y) \bigg(\displaystyle\limsup_{r \rightarrow \infty} \frac{m(r+1)}{r} \bigg)
\end{displaymath}
and
\begin{displaymath}
Mdim(f(x):y) \leq Mdim(x:y) \bigg(\displaystyle\limsup_{r \rightarrow \infty} \frac{m(r+1)}{r} \bigg).
\end{displaymath}
\end{lemma}

\begin{proof}
By Lemma \ref{mod dpi}, we have
\begin{align*}
mdim(f(x):y) &\leq \displaystyle\liminf_{r \rightarrow \infty} \frac{I_{m(r+1)}(x:y)}{r}\\
						 &= \displaystyle\liminf_{r \rightarrow \infty} \bigg(\frac{I_{m(r+1)}(x:y)}{m(r+1)} \cdot \frac{m(r+1)}{r} \bigg)\\
						 &\leq mdim(x:y) \bigg(\displaystyle\limsup_{r \rightarrow \infty} \frac{m(r+1)}{r} \bigg).
\end{align*}
A similar proof can be given for $Mdim$. \qedhere
\end{proof}

Theorem \ref{dpi} follows immediately from Lemma \ref{mod proc lemma} and the following well-known observation.

\begin{observation}\label{lip mod}
A function $f: \mathbb{R}^n \rightarrow \mathbb{R}^k$ is Lipschitz if and only if there exists $s \in \mathbb{N}$ such that $m(r) = r + s$ is a modulus for $f$.
\end{observation}

We can derive a similar observation for H\"{o}lder functions. (Recall the definition of H\"{o}lder functions given in the introduction.)

\begin{observation}\label{hold mod}
If a function $f:\mathbb{R}^n \rightarrow \mathbb{R}^k$ is H\"{o}lder with exponent $\alpha$, then there exists $s \in \mathbb{N}$ such that $m(r) = \lceil \frac{1}{\alpha}(r + s) \rceil$ is a modulus for $f$.
\end{observation}

We can derive the following fact from Observation \ref{hold mod} and the modulus processing lemma. 

\begin{corollary}\label{dpi holder}
If $f:\mathbb{R}^n \rightarrow \mathbb{R}^k$ is computable and H\"{o}lder with exponent $\alpha$, then, for all $x \in \mathbb{R}^n$ and $y \in \mathbb{R}^t$,
\begin{displaymath}
mdim(f(x):y) \leq \frac{1}{\alpha}mdim(x:y)
\end{displaymath}
and
\begin{displaymath}
Mdim(f(x):y) \leq \frac{1}{\alpha}Mdim(x:y).
\end{displaymath}
\end{corollary}
\section{Reverse Data Processing Inequalities}
In this section we develop reverse versions of the data processing inequalities from section 6.
\begin{notation}
Let $n \in \mathbb{Z}^+$.
\begin{enumerate}
\item $[n] = \{1, \cdots, n\}$.
\item For $S \subseteq [n]$, $x \in \mathbb{R}^{|S|}$, $y \in \mathbb{R}^{n-|S|}$, the string
\begin{displaymath}
x *_S y \in \mathbb{R}^n
\end{displaymath}
is obtained by placing the components of $x$ into the positions in $S$ (in order) and the components of $y$ into the positions in $[n] \setminus S$ (in order).
\item For each $x = (x_1, x_2, \ldots, x_n) \in \mathbb{R}^n$, let $x_{(i,j)} = (x_i, x_{i+1}, \ldots, x_j)$ for every $i,j \in \mathbb{N}$ such that $i \leq j \leq n$.
\end{enumerate}
\end{notation}

\begin{definition}
Let $f: \mathbb{R}^n \rightarrow \mathbb{R}^k$.
\begin{enumerate}
\item $f$ is \emph{co-Lipschitz} if there is a real number $c > 0$ such that for all $x,y \in \mathbb{R}^n$,
\begin{align*}
|f(x) - f(y)| \geq c|x - y|.
\end{align*}
\item $f$ is \emph{bi-Lipschitz} if $f$ is both Lipschitz and co-Lipschitz.
\item For $S \subseteq [n]$, $f$ is $S$-\emph{co-Lipschitz} if there is a real number $c > 0$ such that, for all $u, v \in \mathbb{R}^{|S|}$ and $y \in \mathbb{R}^{n-|S|}$,
\begin{align*}
|f(u *_{S} y) - f(v *_{S} y)| \geq c|u - v|.
\end{align*}
\item For $i \in [n]$, $f$ is \emph{co-Lipschitz in its $i^{th}$ argument} if $f$ is $\{i\}$-co-Lipschitz.
\end{enumerate}
\end{definition}

Note that $f$ is $[n]$-co-Lipschitz if and only if $f$ is co-Lipschitz.\\\\
{\bf Example.} The function $f: \mathbb{R}^n \rightarrow \mathbb{R}$ defined by
\begin{displaymath}
f(x_1, \cdots, x_n) = x_1 + \cdots + x_n
\end{displaymath}
is $S$-co-Lipschitz if and only if $|S| \leq 1$. In particular, if $n \geq 2$, then $f$ is co-Lipschitz in every argument, but $f$ is not co-Lipschitz.

We next relate co-Lipschitz conditions to moduli.

\begin{definition}
Let $f: \mathbb{R}^n \rightarrow \mathbb{R}^k$.
\begin{enumerate}
\item An \emph{inverse modulus} for $f$ is a nondecreasing function $m': \mathbb{N} \rightarrow \mathbb{N}$ such that, for all $x,y \in \mathbb{R}^n$ and $r \in \mathbb{N}$,
\begin{displaymath}
|f(x) - f(y)| \leq 2^{-m'(r)} \Rightarrow |x - y| \leq 2^{-r}.
\end{displaymath}
\item Let $S \subseteq [n]$. An $S$-\emph{inverse modulus} for $f$ is a nondecreasing function $m': \mathbb{N} \rightarrow \mathbb{N}$ such that, for all $u,v \in \mathbb{R}^{|S|}$, all $y \in \mathbb{R}^{n-|S|}$, and all $r \in \mathbb{N}$,
\begin{displaymath}
|f(u *_S y) - f(v *_S y)| \leq 2^{-m'(r)} \Rightarrow |u - v| \leq 2^{-r}.
\end{displaymath}
\item Let $i \in [n]$. An \emph{inverse modulus} for $f$ \emph{in its} $i^{th}$ \emph{argument} is an $\{i\}$-inverse modulus for $f$.
\end{enumerate}
\end{definition}

\begin{observation}\label{s co lip mod} Let $f: \mathbb{R}^n \rightarrow \mathbb{R}^k$ and $S \subseteq [n]$.
\begin{enumerate}
\item $f$ is $S$-co-Lipschitz if and only if there is a positive constant $t \in \mathbb{N}$ such that $m'(r) = r + t$ is an $S$-inverse modulus of $f$.
\item $f$ is co-Lipschitz if and only if there is a positive constant $t \in \mathbb{N}$ such that $m'(r) = r + t$ is an inverse modulus of $f$.
\end{enumerate}
\end{observation}

\begin{definition}
Let $f: \mathbb{R}^n \rightarrow \mathbb{R}^t$ and $S \subseteq [n]$. We say that $f$ is $S$-\emph{injective} if, for all $x,y \in \mathbb{R}^n$ and $z \in \mathbb{R}^{n - |S|}$,
\begin{align*}
f(x *_S z) = f(y *_S z) \Rightarrow x = y.
\end{align*}
\end{definition}
Note $f$ is injective if and only if $f$ is $[n]$-injective.

\begin{definition}
Let $f: \mathbb{R}^n \rightarrow \mathbb{R}^t$ be a function and $S \subseteq [n]$ such that $n \in \mathbb{N}$. An \emph{S-left inverse} of $f$ is a partial function $g: \mathbb{R}^t \times \mathbb{R}^{n-|S|} \rightarrow \mathbb{R}^{|S|}$ such that, for all $x \in \mathbb{R}^{|S|}$ and $y \in \mathbb{R}^t \times \mathbb{R}^{n-|S|}$,
\begin{align*}
g(f(x *_S y), y) = x.
\end{align*}
\end{definition}
It is easy to prove that $f$ has an $S$-left inverse if and only if $f$ is $S$-injective.

\begin{lemma}\label{inv mod sinj}
If $f:\mathbb{R}^n \rightarrow \mathbb{R}^t$ has an $S$-inverse modulus $m'$, then $f$ is $S$-injective and $m'$ is a modulus for any $S$-left inverse of $f$.
\end{lemma}

\begin{proof}
Let $m':\mathbb{N} \rightarrow \mathbb{N}$ be an $S$-inverse modulus for $f$, $x,y \in \mathbb{R}^{|S|}$ and $z \in \mathbb{R}^{n - |S|}$, then, if $f(x *_S z) = f(y *_S z)$,
\begin{align*}
|f(x *_S z) - f(y *_S z)| \leq 2^{-m'(r)},
\end{align*}
for all $r \in \mathbb{N}$, which implies that
\begin{align*}
|x - y| \leq 2^{-r}.
\end{align*}
Therefore, $x = y$ and $f$ is $S$-injective.

Let $g: \mathbb{R}^t \times \mathbb{R}^{n - |S|} \rightarrow \mathbb{R}^{|S|}$ be an $S$-left inverse of $f$. Let $x,y \in dom\,g$ and $r \in \mathbb{N}$ such that $x = (f(u *_S w), w)$ and $y = (f(v *_S z), z)$, where $u,v \in \mathbb{R}^{|S|}$ and $w,z \in \mathbb{R}^{n - |S|}$. Assume that $|x - y| \leq 2^{-m'(r)}$, then
\begin{align*}
&\hspace*{5mm}|f(g(f(u *_S w), w) *_S w) - f(g(f(v *_S z), z) *_S z)|\\
&= |f(u *_S w) - f(v *_S z)|\\
&\leq |(f(u *_S w), w) - (f(v *_S z), z)|\\
&= |x - y|\\
&\leq 2^{-m'(r)}.
\end{align*}
So, $|g(f(u *_S w), w) - g(f(v *_S z), z)| \leq 2^{-r}$, and
\begin{align*}
|g(x) - g(y)| &= |g(f(u *_S w), w) - g(f(v *_S z), z)|\\
							&\leq 2^{-r}.
\end{align*}
Therefore, $m'$ is a modulus for $g$. \qedhere
\end{proof}

\begin{lemma}\label{comp left inv}
If $f: \mathbb{R}^n \rightarrow \mathbb{R}^t$ is a computable and uniformly continuous function that has a computable $S$-inverse modulus $m'$, then $f$ has a computable $S$-left inverse.
\end{lemma}

\begin{proof}
Assume the hypothesis. Since $f$ is computable and uniformly continuous, there exist a modulus $m$ for $f$ and an oracle Turing machine $M_f$ such that, for every $x \in \mathbb{R}^n$, $r \in \mathbb{N}$, and every oracle $h_x$ for $x$, 
\begin{align}\label{71}
|M_f^{h_x}(r) - f(x)| \leq 2^{-r}.
\end{align}
Define $g: \mathbb{R}^t \times \mathbb{R}^{n-|S|} \rightarrow \mathbb{R}^{|S|}$ by
\begin{align*}
g(z) = \twopartdef{x}{z = (f(x *_S y),y),}{\textnormal{undefined}}{\textnormal{otherwise}},
\end{align*}
where $x \in \mathbb{R}^{|S|}$, $y \in \mathbb{R}^{n - |S|}$, and $z \in \mathbb{R}^t \times \mathbb{R}^{n-|S|}$.

We now show that $g$ is computable. Let $z = (f(x *_S y), y) \in dom\,g$ and $h_z$ be an oracle for $z$ such that, for all $r \in \mathbb{N}$,
\begin{align}\label{72}
|h_z(r) - z| \leq 2^{-r}.
\end{align}
First we show that, for any $r \in \mathbb{N}$, there exist a rational $q \in \mathbb{Q}^{|S|}$ and an oracle $h_{qy}$ for $q *_S y$ such that
\begin{align*}
|M_f^{h_{qy}}(m'(r) + 3) - h_z(m'(r) + 3)| \leq 2^{-(m'(r) + 1)}.
\end{align*}
Let $q \in \mathbb{Q}^{|S|}$ such that $|q *_S y - x *_S y| \leq 2^{-(m(m'(r) + 2))}$, and let $h_{qy}(r) = q *_S h_z(r)_{(t+1, t+n-|S|)}$ be an oracle for $q *_S y$. Therefore,
\begin{align}\label{73}
|f(q *_S y) - f(x *_S y)| \leq 2^{-(m'(r) + 2)}.
\end{align}
By (\ref{71}), (\ref{72}), (\ref{73}),
\begin{align*}
&\,\,\,\,\,\,\,\,|M_f^{h_{qy}}(m'(r)+3) - h_z(m'(r) + 3)_{(1,t)}|\\
&=|M_f^{h_{qy}}(m'(r) + 3) - f(q *_{S} y) + f(q *_{S} y) - f(x *_S y) + f(x *_S y) - h_z(m'(r) + 3)_{(1,t)}|\\
&\leq |M_f^{h_{qy}}(m'(r) + 3) - f(q *_{S} y)| + |f(q *_{S} y) - f(x *_S y)| + |h_z(m'(r) + 3)_{(1,t)} - f(x *_S y)|\\
&\leq 2^{-(m'(r) + 3)} + 2^{-(m'(r) + 2)} + 2^{-(m'(r) + 3)}\\
&= 2^{-(m'(r) + 1)}.
\end{align*}
Let $M_g$ be a Turing machine equipped with oracle $h_z$. Given an input $r \in \mathbb{N}$, $M_g$ searches for and outputs a rational $q_x \in \mathbb{Q}^{|S|}$ such that
\begin{align}\label{74}
|M^{h_{q_xy}}_f(m'(r) + 3) - h_z(m'(r) + 3)_{(1,t)}| \leq 2^{-(m'(r) + 1)},
\end{align}
where $h_{q_xy} = q_x *_S h_z(r)_{(t+1, t+n-|S|)}$ is an oracle for $q_x *_S y$. We now show that $|M^{h_z}_g(r) - g(z)| \leq 2^{-r}$. By (\ref{71}), (\ref{72}), (\ref{74}),
\begin{align*}
&\,\,\,\,\,\,\,\,|f(M_g^{h_z}(r) *_S y) - f(x *_S y)|\\
&= |f(q_x *_S y) - f(x *_S y)|\\
&= |f(q_x *_S y) - M_f^{h_{q_xy}}(m'(r) + 3) + M_f^{h_{q_xy}}(m'(r) + 3) - h_z(m'(r) + 3)_{(1,t)} + h_z(m'(r) + 3)_{(1,t)} - f(x *_S y)|\\
&\leq |f(q_x *_S y) - M_f^{h_{q_xy}}(m'(r) + 3)| + |M_f^{h_{q_xy}}(m'(r) + 3) - h_z(m'(r) + 3)_{(1,t)}| + |h_z(m'(r) + 3)_{(1,t)} - f(x *_S y)|\\
&\leq 2^{-(m'(r) + 3)} + 2^{-(m'(r) + 1)} + 2^{-(m'(r) + 3)}\\
&= 2^{-(m'(r) + 2)} + 2^{(m'(r) + 1)}\\
&< 2^{-m'(r)}.
\end{align*}
Since $m'$ is an $S$-inverse modulus for $f$, we have
\begin{align*}
|M_g^{h_z}(r) - g(z)| &= |M_g^{h_z}(r) - x|\\
											&\leq 2^{-r}.
\end{align*}
Therefore, $g$ is a computable $S$-left inverse of $f$.
\end{proof}
\begin{lemma}[reverse modulus processing lemma]\label{uc rev dpi}
If $f: \mathbb{R}^n \rightarrow \mathbb{R}^k$ is a computable and uniformly continuous function, and $m'$ is a computable $S$-inverse modulus for $f$, then, for all $S \subseteq [n]$, $x \in \mathbb{R}^{|S|}$, $y \in \mathbb{R}^t$, and $z \in \mathbb{R}^{n-|S|}$,
\begin{displaymath}
mdim(x:y) \leq mdim((f(x *_S z),z):y) \bigg(\displaystyle\limsup_{r \rightarrow \infty} \frac{m'(r+1)}{r} \bigg)
\end{displaymath}
and
\begin{displaymath}
Mdim(x:y) \leq Mdim((f(x *_S z),z):y) \bigg(\displaystyle\limsup_{r \rightarrow \infty} \frac{m'(r+1)}{r} \bigg).
\end{displaymath}
\end{lemma}

\begin{proof}
Assume the hypothesis. By Lemmas \ref{inv mod sinj} and \ref{comp left inv}, there exists a computable and uniformly continuous function $g$ that is an $S$-left inverse of $f$ and $m'$ is a modulus for $g$. Then, for all $S \subseteq [n]$, $x \in \mathbb{R}^{|S|}$, $y \in \mathbb{R}^t$, and $z \in \mathbb{R}^{n - |S|}$,
\begin{displaymath}
mdim(x:y) = mdim(g(f(x *_S z),z):y).
\end{displaymath}
Therefore, by Lemma \ref{mod proc lemma}, we have
\begin{displaymath}
mdim(x:y) \leq mdim(f(x *_S z),z:y) \bigg(\displaystyle\limsup_{r \rightarrow \infty} \frac{m'(r+1)}{r} \bigg).
\end{displaymath}
A similar proof can be given for $Mdim$. \qedhere
\end{proof}

By Observation \ref{s co lip mod} and Lemma \ref{uc rev dpi}, we have the following.

\begin{theorem}[reverse data processing inequality]\label{rev dpi}
If $S \subseteq [n]$ and $f:\mathbb{R}^n \rightarrow \mathbb{R}^k$ is computable and $S$-co-Lipschitz, then, for all $x \in \mathbb{R}^{|S|}$, $y \in \mathbb{R}^t$, and $z \in \mathbb{R}^{n-|S|}$,
\begin{displaymath}
mdim(x:y) \leq mdim((f(x *_S z),z):y)
\end{displaymath}
and
\begin{displaymath}
Mdim(x:y) \leq Mdim((f(x *_S z),z):y).
\end{displaymath}
\end{theorem}

\begin{definition}
Let $f: \mathbb{R}^n \rightarrow \mathbb{R}^k$ and $0 < \alpha \leq 1$.
\begin{enumerate}
\item $f$ is \emph{co-H\"{o}lder with exponent} $\alpha$ if there is a real number $c > 0$ such that, for all $x,y \in \mathbb{R}^n$,
\begin{align*}
|x - y| \leq c|f(x) - f(y)|^{\alpha}.
\end{align*}
\item For $S \subseteq [n]$, $f$ is $S$-\emph{co-H\"{o}lder with exponent} $\alpha$ if there is a real number $c > 0$ such that, for all $u, v \in \mathbb{R}^{|S|}$ and $y \in \mathbb{R}^{n-|S|}$,
\begin{align*}
|u - v| \leq c|f(u *_{S} y) - f(v *_{S} y)|^{\alpha}.
\end{align*}
\end{enumerate}
\end{definition}

\begin{observation}\label{s co hold mod} Let $f: \mathbb{R}^n \rightarrow \mathbb{R}^k$ and $S \subseteq [n]$.
\begin{enumerate}
\item If $f$ is $S$-co-H\"{o}lder with exponent $\alpha$, then there exists $t \in \mathbb{N}$ such that $m'(r) = \lceil \frac{1}{\alpha}(r + t) \rceil$ is an $S$-inverse modulus of $f$.
\item If $f$ is co-H\"{o}lder with exponent $\alpha$, then there exists $t \in \mathbb{N}$ such that $m'(r) = \lceil \frac{1}{\alpha}(r + t) \rceil$ is an inverse modulus of $f$.
\end{enumerate}
\end{observation}

The next corollary follows from the reverse modulus processing lemma and Observation \ref{s co hold mod}.

\begin{corollary}\label{rev dpi s co hold}
If $S \subseteq [n]$ and $f:\mathbb{R}^n \rightarrow \mathbb{R}^k$ is computable and $S$-co-H\"{o}lder with exponent $\alpha$, then, for all $x \in \mathbb{R}^{|S|}$, $y \in \mathbb{R}^t$, and $z \in \mathbb{R}^{n-|S|}$,
\begin{displaymath}
mdim(x:y) \leq \frac{1}{\alpha}mdim((f(x *_S z),z):y)\\
\end{displaymath}
and
\begin{displaymath}
Mdim(x:y) \leq \frac{1}{\alpha}Mdim((f(x *_S z),z):y).
\end{displaymath}
\end{corollary}
\section{Data Processing Applications}
In this section we use the data processing inequalities and their reverses to investigate how certain functions on Euclidean space preserve or predictably transform mutual dimensions.

\begin{theorem}[mutual dimension conservation inequality]\label{conservation} If $f: \mathbb{R}^n \rightarrow \mathbb{R}^k$ and $g:\mathbb{R}^t \rightarrow \mathbb{R}^l$ are computable and Lipschitz, then, for all $x \in \mathbb{R}^n$ and $y \in \mathbb{R}^t$,
\begin{displaymath}
mdim(f(x):g(y)) \leq mdim(x:y)
\end{displaymath}
and
\begin{displaymath}
Mdim(f(x):g(y)) \leq Mdim(x:y).
\end{displaymath}
\end{theorem}

\begin{proof}
The conclusion follows from Theorem \ref{mutual dim properties} and the data processing inequality.
\begin{align*}
mdim(f(x):g(y)) &\leq mdim(x:g(y))\\
								&= mdim(g(y):x)\\
								&\leq mdim(y:x)\\
								&=mdim(x:y).
\end{align*}
A similar argument can be given for $Mdim(f(x):g(y)) \leq Mdim(x:y)$.
\qedhere
\end{proof}

\begin{theorem}[reverse mutual dimension conservation inequality]\label{rev cons s co lip}
Let $S_1 \subseteq [n]$ and $S_2 \subseteq [t]$. If $f: \mathbb{R}^n \rightarrow \mathbb{R}^k$ is computable and $S_1$-co-Lipschitz, and $g:\mathbb{R}^t \rightarrow \mathbb{R}^l$ is computable and $S_2$-co-Lipschitz, then, for all $x \in \mathbb{R}^{|S_1|}$, $y \in \mathbb{R}^{|S_2|}$, $w \in \mathbb{R}^{n - |S_1|}$, and $z \in \mathbb{R}^{t - |S_2|}$,
\begin{displaymath}
mdim(x:y) \leq mdim((f(x *_S w),w):(g(y *_S z),z))
\end{displaymath}
and
\begin{displaymath}
Mdim(x:y) \leq Mdim((f(x *_S w),w):(g(y *_S z),z)).
\end{displaymath}
\end{theorem}

\begin{proof}
The conclusion follows from Theorem \ref{mutual dim properties} and the reverse data processing inequality.
\begin{align*}
mdim(x:y) &\leq mdim((f(x *_S w),w):y)\\
								&= mdim(y:(f(x *_S w),w))\\
								&\leq mdim((g(y *_S z),z):(f(x *_S w),w))\\
								&= mdim((f(x *_S w),w):(g(y *_S z),z)).
\end{align*}
A similar argument can be given for $Mdim(x:y) \leq Mdim((f(x *_S w),w):(g(y *_S z),z))$.
\qedhere
\end{proof}

\begin{corollary}[preservation of mutual dimension]
If $f: \mathbb{R}^n \rightarrow \mathbb{R}^k$ and $g: \mathbb{R}^t \rightarrow \mathbb{R}^l$ are computable and bi-Lipschitz, then, for all $x \in \mathbb{R}^n$ and $y \in \mathbb{R}^t$,
\begin{displaymath}
mdim(f(x):g(y)) = mdim(x:y)
\end{displaymath}
and
\begin{displaymath}
Mdim(f(x):g(y)) = Mdim(x:y).
\end{displaymath}
\end{corollary}

\begin{corollary}\label{conservation hold}
If $f: \mathbb{R}^n \rightarrow \mathbb{R}^k$ and $g:\mathbb{R}^t \rightarrow \mathbb{R}^l$ are computable and H\"{o}lder with exponents $\alpha$ and $\beta$, respectively, then, for all $x \in \mathbb{R}^n$ and $y \in \mathbb{R}^t$,
\begin{displaymath}
mdim(f(x):g(y)) \leq \frac{1}{\alpha\beta}mdim(x:y)
\end{displaymath}
and
\begin{displaymath}
Mdim(f(x):g(y)) \leq \frac{1}{\alpha\beta}Mdim(x:y).
\end{displaymath}
\end{corollary}

\begin{corollary}\label{rev cons s co hold}
Let $S_1 \subseteq [n]$ and $S_2 \subseteq [t]$. If $f: \mathbb{R}^n \rightarrow \mathbb{R}^k$ is computable and $S_1$-co-H\"{o}lder with exponent $\alpha$, and $g:\mathbb{R}^t \rightarrow \mathbb{R}^l$ is computable and $S_2$-co-H\"{o}lder with exponent $\beta$, then, for all $x \in \mathbb{R}^{|S_1|}$, $y \in \mathbb{R}^{|S_2|}$, $w \in \mathbb{R}^{n - |S_1|}$, and $z \in \mathbb{R}^{t - |S_2|}$,
\begin{displaymath}
mdim(x:y) \leq \frac{1}{\alpha\beta}mdim((f(x *_S w),w):(g(y *_S z),z))
\end{displaymath}
and
\begin{displaymath}
Mdim(x:y) \leq \frac{1}{\alpha\beta}Mdim((f(x *_S w),w):(g(y *_S z),z)).
\end{displaymath}
\end{corollary}
\section{Conclusion}
We expect mutual dimensions and the data processing inequalities to be useful for future research in computable analysis. We also expect the development of mutual dimensions in Euclidean spaces --- highly structured spaces in which it is clear that $mdim$ and $Mdim$ are the right notions --- to guide future explorations of mutual information in more challenging contexts.
\\\\
{\bf Acknowledgments.}
We thank Elvira Mayordomo, A. Pavan, Giora Slutzki, and Jim Lathrop for useful discussions. We also thank three anonymous reviewers for useful suggestions.
\bibliography{Master}
\end{document}